\documentclass[a4paper,11pt,times,authoryears,reqno]{amsart}

\usepackage{geometry}
\usepackage{bm}
\usepackage{amssymb,amsfonts}
\usepackage{graphicx}
\usepackage{mathrsfs}
\usepackage{subcaption}
\usepackage{epstopdf}
\usepackage{natbib}
\usepackage{enumerate}
\usepackage{booktabs}
\usepackage{adjustbox}
\usepackage{pdflscape}
\usepackage[colorlinks,citecolor=blue]{hyperref}

\newtheorem{theorem}{Theorem}[section]

\theoremstyle{definition}

\theoremstyle{remark}

\newtheorem{proposition}[theorem]{Proposition}

\numberwithin{equation}{section}
\numberwithin{equation}{section}

\numberwithin{equation}{section}

\def\bW{{\mathbf W}}

\def\bA{{\mathbf A}}  \def\bD{{\mathbf D}}

 \def\bW{{\mathbf W}} \def\bX{{\mathbf X}} 

\def\bM{\boldsymbol M}

\def\bx{{\mathbf x}} \def\by{{\mathbf y}} 

\def\bbeta{{\boldsymbol{\beta}}}

 \def\bdelta{{\boldsymbol{\delta}}}

 \def\bGamma{{\boldsymbol{\Gamma}}}

  \def\bxi{{\boldsymbol{\xi}}}
       
\def\bSigma{{\boldsymbol{\Sigma}}}

 \def\tbbeta{\widetilde{\boldsymbol \beta}}

\def\hbbeta{\widehat{\boldsymbol \beta}}

\def\sgn{\mbox{sgn}}

\def\real{\mathop{{\rm I}\kern-.2em\hbox{\rm R}}\nolimits}

\def\1overn{\frac{1}{n}}

\def\bel{\begin{eqnarray}\label}  \def\eel{\end{eqnarray}}
\def\bes{\begin{eqnarray*}}  \def\ees{\end{eqnarray*}}

\begin{document}

\title{Improving Estimations in Quantile Regression Model with 
Autoregressive Errors}


\author{Bahad{\i}r Y\"{u}zba\c{s}{\i}$^1$, Yasin Asar$^2$, M.\c{S}amil \c{S}{\i}k$^3$ \and Ahmet Demiralp$^4$ }

\date{\today}
\maketitle

{\footnotesize
\center { \text{  $^{1,3,4}$ Department of Econometrics}\par
  { \text{ Inonu University}}\par
  {\text{ Malatya 44280, Turkey}}\par
  { \texttt{E-mail address: $^1$b.yzb@hotmail.com and $^3$mhmd.sml85@gmail.com and $^4$ahmt.dmrlp@gmail.com}}\par

  \vskip 0.2 cm

  \text{  $^2$ Department of Mathematics-computer Sciences
}\par
  { \text{ Necmettin Erbakan University}}\par
  {\text{ Konya 42090, Turkey}}\par
  { \texttt{E-mail address: yasar@konya.edu.tr, yasinasar@hotmail.com}}

}}


\renewcommand\leftmark {\centerline{  \rm Quantile Shrinkage Estimation}}
\renewcommand\rightmark {\centerline{ \rm  Quantile Shrinkage Estimation}}

\renewcommand{\thefootnote}{}

\footnote {\par

Corresponding author : Bahad{\i}r Y\"{u}zba\c{s}{\i} \par

}

\begin{abstract}

An important issue is that the respiratory mortality may be a result of air pollution which can be measured by the following variables: temperature, relative humidity, carbon monoxide, sulfur dioxide, nitrogen dioxide, hydrocarbons, ozone and particulates. The usual way is to fit a model using the ordinary least squares regression, which has some assumptions, also known as Gauss-Markov assumptions, on the error term showing white noise process of the regression model. However, in many applications, especially for this example, these assumptions are not satisfied. Therefore, in this study, a quantile regression approach is used to model the respiratory mortality using the mentioned explanatory variables. Moreover, improved estimation techniques such as preliminary testing and shrinkage strategies are also obtained when the errors are autoregressive. A Monte Carlo simulation experiment, including the quantile penalty estimators such as Lasso, Ridge and Elastic Net, is designed to evaluate the performances of the proposed techniques. Finally, the theoretical risks of the listed estimators are given.

\end{abstract}

\maketitle

\section{Introduction}
Regression analysis is a statistical technique that  is used to model the cumulative and linear relationship between covariates and response variables. The most common method used for this purpose is the ordinary least squares (OLS) method. The linear regression model can be written as follows:
\begin{equation}
y_i=\beta_0+\sum_{j=1}^{p}\beta_jx_{ij} +\varepsilon_i, \\ \\
i=1,2,\dots,n,\label{lin.mod}
\end{equation}%
where $y_i$'s are the response variables, $\beta _{j}$'s are unknown regression coefficients, $x_{ij}$'s are known covariates and  $\varepsilon_{i}$'s are unobservable random errors. When estimating the parameters using OLS method, the expectation of the dependent variable conditional on the independent variables is obtained. In other words, the relationship between the explanatory and explained variables in the coordinate plane is estimated with a mean regression line.


 In order to use OLS estimator, there are three assumptions on the error terms showing white noise process of the regression model: (1) The error terms have zero mean, (2) The variance of the error terms is constant and (3) The covariance between the errors is zero i.e., there is no autocorrelation problem. In real life most of the data doesn't provide these assumptions. Moreover, OLS provides a view of the relationship between covariates and response variable such that it models the expectation of the response conditional on the covariates without taking into account the outliers. To overcome these inadequacies of the classical regression \cite{qr1978} have proposed the quantile regression as an expansion of the classical regression model to a basic minimization problem which generates sample quantiles. For a random variable $Y$  with distribution function $\mathcal{F}_{Y}\left(\by\right)=P( Y\le\by ) =\tau$ and
$0\le\tau\le 1$, the $\tau^{th}$ quantile function of $Y$, $\mathcal{Q}_{\tau}(\by)$, is defined to be
\begin{equation}\label{eq:y_tau}
\mathcal{Q}_{\tau}(Y\vert \bX)=\by_{\tau}=\mathcal{F}_{Y}^{-1}(\tau)=\inf\left\{\by : \mathcal{F}_{Y}\left(\by\right)\ge\tau\right\}\equiv \bx_i'\bbeta_\tau 
\end{equation}
where $\by_{\tau}$ is the inverse function of $\mathcal{F}_{Y}\left(\tau\right)$ for the $\tau$th quantile, $\by=\left(y_{1},y_{2},\dots,y_{n}\right)$, $\bX=(\bx_1,\dots,\bx_n)'$ and $\bx_i=\left(
x_{i1},x_{i2},\dots,x_{ip}\right)'$. In other words, the $\tau^{th}$ quantile in a sample corresponds to the probability $\tau$ for a $\by$ value. Also an estimation of the full model (FM) $\tau^{th}$  quantile regression coefficients can be defined by solving the following minimization of problem 
\begin{equation}
\hbbeta_{\tau} = \underset{\bbeta\in\Re^{p}}{\arg \min} \sum_{i=1}^{n}\rho_\tau (y_i-\bx'_{i}\bbeta),
\end{equation}
where $\rho_\tau(u)=u(\tau-I(u<0))$ is the quantile loss function. Hence, it yields
\begin{equation}\label{eq:est_tau}
\hbbeta_{\tau}=\underset{\bbeta\in\Re^{p}}{\arg \min}\Bigg[ \sum_{i\in\left\{i: y_i\ge \bx_i^{'}\bbeta\right\}}^{n}
\tau\vert\ y_i-\bx'_{i}\bbeta\vert - \sum_{i\in\left\{i: y_i\ <\bx'_i\bbeta\right\}}^{n}
(1-\tau)\vert\ y_i-\bx'_{i}\bbeta\vert\Bigg].
\end{equation}

\citet{Ko-Xi2006} proposed a quantile autoregression (QAR) model which could be interpreted as a special case of the general random-coefficient autoregression model with strongly dependent coefficients. The authors studied statistical properties of the proposed model and associated estimators and derived the limiting distributions of the autoregression quantile process. \cite{Koenker2008} proposed the \textit{quantreg} R package and it is implementations for linear, non-linear and non-parametric quantile regression models.  The R software and the package \textit{quantreg} are open-source software projects and can be freely downloaded from CRAN: \url{http://cran.r-project.org}. \cite{ciric2012} compare different computational intelligence methodologies based on artificial neural networks used for forecasting an air quality parameter. \citet{tang2015} proposed composite quantile regression for dependent data. The authors also showed the root-n consistency and asymptotic normality of the composite quantile estimator. Moreover, the authors apply their proposed method to NO$_2$ particle data in which air pollution on a road is modeled via traffic volume and meteorological variables. \citet{wang-lin2015} proposed a penalized quantile estimator in semiparametric linear regression model and dealt with longitudinal data. The authors obtained the oracle properties of the estimator and selection consistency.

The books by Koenker (\cite{Koenker2005}) and Davino (\cite{Davino2014}) are excellent sources for various properties of Quantile Regression as well as many computer algorithms. Moreover, \cite{hqreg} developed an algorithm, called semismooth Newton coordinate descent (SNCD), to obtain a better efficiency and scalability for computing the solution paths of penalized quantile regression. They also provide an R package called \textit{hqreg}. Moreover, this package also obtains Lasso of (\citet{lasso}), Ridge of (\citet{Ho-Ke1970}) and Elastic Net of (\citet{elastic-net}) estimators in the quantile regression models. The \textit{hqreg} functions give the solution path while the  \textit{quantreg} package of \citet{Koenker2013} computes a single solution.

On the other hand, the book of \cite{ahmed2014} can be found the large literature and informations about shrinkage estimations in the context of linear and partially linear models (PLMs). The preliminary and Stein-type estimations based on ridge regression are obtained by \cite{yuzbasi-et-al2017} for linear models and by \cite{yuzbasi-ahmed2016} for PLMs. Furthermore, \cite{yuzbasi-et-al2017b,yuzbasi-et-al2017c} introduced the pretest and shrinkage estimation based on the quantile regression when the errors are both i.i.d. and non-i.i.d, respectively. In these studies, asymptotic distributional bias, quadratic bias and risk functions are also obtained. The novelty of this study is the errors having the problem of autocorrelation which is very common in time series analysis.

The paper is organized as follows: in Section~\ref{sec:ME}, we consider a real data example in order to examine the assumptions of the classical linear regression. The pretest, shrinkage estimators and penalized estimations are also given in Section \ref{sec:SM}. Also, we estimate the listed estimators in Section \ref{sec:MEC}. The design and the results of a Monte Carlo simulation study including a comparison with other penalty estimators are given in Section ~\ref{sim}. The asymptotic distributional risk properties of the pretest and shrinkage estimators are obtained in Section~\ref{sec:TR}. Finally,  the concluding remarks are presented in Section ~\ref{conc}.

\section{Motivation Example}
\label{sec:ME}

In this section, we consider the study of \cite{Shumway1998} of the possible effects of temperature and pollution on weekly mortality in Los Angeles (LA) Country. This data has $508$ weekly observations from $1970$ to $1979$. In Table~\ref{Tab:variables:lap}, we describe the variables of the cement data which is freely available in the {\it astsa} package with the function {\it lap} in R project.

\begin{table}[!htbp]
\centering
\begin{tabular}{ll}
\toprule
\textbf{Variables} & \textbf{Descriptions} \\
\midrule
\textbf{Response Variable} &\\
rmort             & Respiratory Mortality  \\
\midrule
\textbf{Predictors}   \\
tempr            & Temperature \\
rh               & Relative Humidity \\
co           & Carbon Monoxide \\
so2             & Sulfur Dioxide \\
no2 & Nitrogen Dioxide \\
hycarb & Hydrocarbons \\
o3   & Ozone \\
part             & Particulates \\
\bottomrule
\end{tabular}
\caption{Descriptions of variables for the LA Pollution-Mortality data set
\label{Tab:variables:lap}}
\end{table}

The Figure \ref{fig:res} shows that the observations $152,153$ and $155$ may be outliers. Applying {\it outlierTest} function in the {\it car} package in R, according to the results, we observe that the observations $152-155$ and $260$ are outliers. We also observe that the errors follow a heavy-tailed distribution.

\begin{figure}[!htbp]
\centering
   \includegraphics[height=8cm,width=12cm]{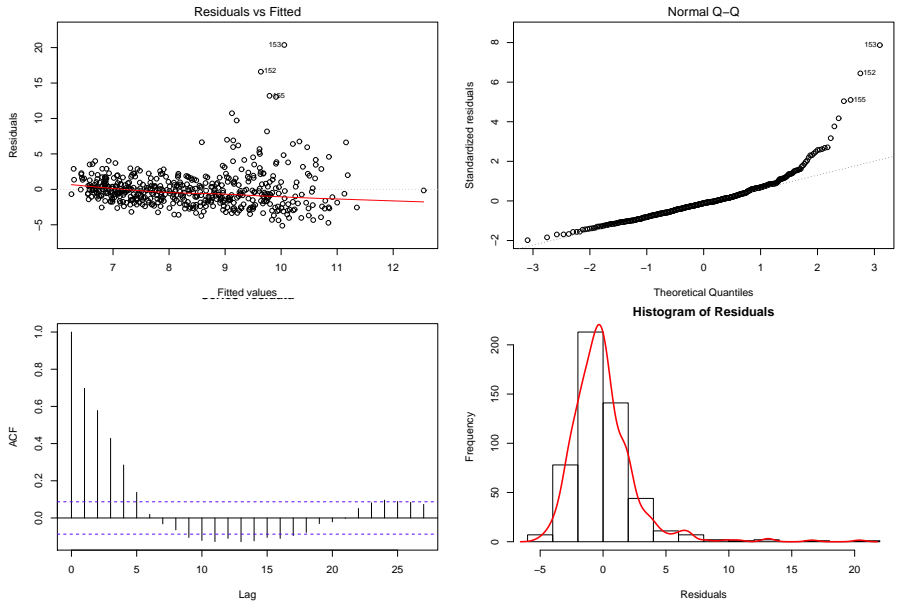}
   \caption{Residual diagnostics}
   \label{fig:res} 
\end{figure}

\begin{table}[!htbp]
\centering
\begin{tabular}{cccc}
\toprule
lag & Autocorrelation& D-W Statistic& p-value \\
\midrule
   1   &   0.697   &   0.604  & 0.000 \\
   2   &   0.578   &   0.841  & 0.000\\
   3   &   0.428   &   1.140  & 0.000\\
   4   &   0.285   &   1.427  & 0.000\\
   5   &   0.138   &   1.719  & 0.006\\
   6   &   0.019   &   1.955  & 0.714\\
\bottomrule
\end{tabular}
\caption{Durbin Watson test
\label{Tab:variables:DW}}
\end{table}

According to Figure \ref{fig:res} and Table \ref{Tab:variables:DW}, the residuals of this data have AR(5) process. Also, we consider the values of $d_L$ and $d_U$ as $1.686$ and $1.852$ respectively. Hence, there is a positive autocorrelation problem on this data.

\begin{table}[!htbp]
\centering
\begin{tabular}{rrrrrrrrr}
  \toprule
 & tempr & rh & co & so2 & no2 & hycarb & o3 & part \\ 
  \midrule
VIF  & 5.197 & 1.673 & 7.711 & 2.636 & 7.377 & 6.071 & 5.698 & 5.360 \\ 
   \bottomrule
\end{tabular}
\caption{VIF values}
\label{TAb:vif}
\end{table}

The values of Table \ref{TAb:vif} and  the ratio of largest eigenvalue to smallest eigenvalue of design matrix in model \eqref{lin.mod} is approximately 657.177 which shows that there is a strong multicollinearity between independent variables. When we consider all results, this dataset suffers from the problems of multicollinearity, autocorrelation, being heavy tailed and having outliers simultaneously. Hence, we will use the quantile type estimation for this data.

\section{Statistical Model}
\label{sec:SM}

Linear regression model in $\eqref{lin.mod}$ would be written in a partitioned form as follows
\begin{equation}
y_i=\bx_{1i}'\bbeta_1+\bx_{2i}'\bbeta_2+\varepsilon _{i},\ \ \
i=1,2,\dots,n,  \label{part.lin}
\end{equation}
where $\bbeta = \left(\bbeta'_1,\bbeta'_2\right)'$ is partitioned so that the coefficient vector of $\bbeta_1=\left( \beta _{1},\beta _{2},\dots,\beta _{p_1}\right)'$, of order $p_1$, is our main interest  and the coefficient vector of $\bbeta_2=\left( \beta _{p_1+1},\beta _{p_1+2},\dots,\beta _{p}\right)'$ is the ``irrelevant variables" with dimension $p_2$, where $p=p_1+p_2$. Also, $\bx_i=\left(\bx_{1i}',\bx_{2i}'\right)$ and $\varepsilon _{i}$ are errors with the same joint distribution function $\mathcal{F}$. The conditional quantile function of response variable $y_i$  can be written as follows
\begin{equation}\label{eq:x_tau}
\mathcal{Q}_{\tau}(y_i\vert \bx_i)=\bx_{1i}'\bbeta_{1,\tau}+\bx_{2i}'\bbeta_{2,\tau},\ \ \ 0<\tau< 1 
\end{equation}
In this study, the main interest is to improve the performance of the important covariates under the following the null hypothesis
\begin{equation}
\label{NullHp}
H_0:\bbeta_{2,\tau}= \bold{0}_{p_2}.
\end{equation}%
If the Equation \eqref{NullHp} is true, then the sub-model (SM) quantile regression estimator of $\bbeta_{\tau}$ is given by $\tbbeta_{\tau}=\left(\tbbeta_{1,\tau},\bm0_{p_2}\right)$, where $
\tbbeta_{1,\tau}=\underset{\bbeta_1\in\Re^{p_1}}{\min} \sum_{i=1}^{n}\rho_\tau (y_i-\bx_{1i}'\bbeta_1)$.

The distribution function ${\mathcal{F}_i}$ is absolutely continuous, with continuous densities $f_i(\xi)$ uniformly bounded away from $0$ and $\infty$ at the points $\xi_i(\tau)$, $i=1,2,...$
\begin{enumerate}
\item[(i)] $\lim_{n\rightarrow\infty} \frac{1}{n}\sum_{i=1}^n\ \bx_i\bx_i'=\bD,\ \ \
\bD_{0}=\frac{1}{n}\bX'\bX$
\item[(ii)] $\lim_{n\rightarrow\infty} \frac{1}{n}\sum_{i=1}^n\ f_i(\bxi_i(\tau))\bx_i\bx_i'=\bD_1$
\item[(iii)] $\max_{1\leq i\leq n}\Vert{\bx_i}\Vert/\sqrt{n}\rightarrow\ 0$
\end{enumerate}
where $\bD_0$ and $\bD_1$ are positive definite matrices.

\subsection{Pretest and Stein-Type Estimations}
The pretest was firstly applied by \citet{Bancroft1944} for the validity of the unclear preliminary information $(\rm UPI)$ by subjecting it to a preliminary test. The pretest estimator $\left(\rm PT\right)$ could be obtained by following equation
\begin{equation}\label{eq:pre}
\hbbeta_{\tau}^{\rm PT}=\hbbeta_{\tau}-\left(\hbbeta_{\tau}-\tbbeta_{\tau}\right)\textrm{I}\left(\mathcal{W}<c_{n,\alpha}\right)
\end{equation}
where $\textrm{I}\left(\cdot\right)$ is an indicator function of a set and $c_{n,\alpha}$ is the $100\left(1-\alpha\right)$ percentage point of the $\mathcal{W}$. In order to test \eqref{NullHp}, under the above assumptions, we consider the following Wald test statistics
\begin{equation}\label{wald}
\mathcal{W}=nw^{-2}\hbbeta_{2,\tau}'\bGamma_{22.1}\hbbeta_{2,\tau}
\end{equation}
where $\bGamma=
\left(
\begin{array}{cc}
\bGamma_{11} & \bGamma_{12} \\
\bGamma_{21} & \bGamma_{22}%
\end{array}%
\right)
=\bD^{-1}\bA\bD^{-1}$, $\bA=\lim_{n\rightarrow\infty}\frac{1}{n}\sum\limits_{i}\sum\limits_{j}\psi{(e_i)}\psi{(e_j)}\bx_i\bx_j^{'}$, the median $\psi{(e_i)}=\sgn{(e_i)}$ and $\bGamma_{22.1} = \bGamma_{22}-\bGamma_{21}\bGamma_{11}^{-1}\bGamma_{12}$. Under the null hypothesis, the distribution of $\mathcal{W}$  follows the chi-square distribution with $p_2$ degree of freedom.

The Stein-type shrinkage (S) estimator is a combination of the over--fitted model estimator $\hbbeta_{\tau}$ with the under--fitted estimator $\tbbeta_{\tau}$, given by
\begin{equation*}
\bm{\widehat{\beta}}_{\tau}^{\textrm{S}}=\hbbeta_{\tau}-d\left( \hbbeta_{\tau}-\tbbeta_{\tau}\right)\mathcal{W}_{n}^{-1} \text{, } d=(p_2-2)\geq 3,
\end{equation*}

In an effort to avoid the over-shrinking problem inherited by $\bm{\widehat{\beta}}_{\tau}^{\textrm{S}}$, we suggest using the positive part of the shrinkage (PS) estimator defined by
\begin{eqnarray*}
\hbbeta_{\tau}^{\textrm{PS}}
&=&\bm{\widehat{\beta}}_{\tau}^{\textrm{S}}-\left( \hbbeta_{\tau}-\tbbeta_{\tau}\right)
\left( 1-d\mathcal{W}_{n}^{-1}\right)\textrm{I}\left(\mathcal{W}_{n}\leq  d\right).
\end{eqnarray*}

\subsection{Quantile Penalized Estimation} We briefly mention about the penalized estimators, given by \cite{hqreg} in quantile regression in a general form as follows:
\begin{equation}
\hbbeta^{\rm Penalized}_{\tau}=\underset{\bm{\bm{\beta}}}{\arg \min}\sum_{i}\rho(y_i-\bx_{i}'\bbeta)+\lambda\ P(\bbeta)
\end{equation}
where $\rho$ is a quantile loss function, $P$ is a penalty function and $\lambda$ is a tuning parameter. Also, 
\begin{equation*}
P(\bbeta)\equiv P_{\alpha}(\bbeta)=\alpha\Vert\bbeta\Vert_1+\frac{(1-\alpha)}{2}\Vert\bbeta\Vert_2^2
\end{equation*}
which is the Lasso penalty for $\alpha=1$ (\citet{lasso}), the Ridge penalty for $\alpha=0$ (\citet{Ho-Ke1970}) and the Elastic-net penalty for $0\le \alpha\le {1}$ (\citet{elastic-net}).

\section{Motivation Example Cont.}
\label{sec:MEC}

In order to apply the proposed methods, we use a two step approach as follows:
\begin{enumerate}
\item[Step 1:] A set of covariates are selected based on a suitable model selection technique since the prior information is not available here.
\item[Step 2:] The full and sub-model estimates are combined in such a way that minimizes the quadratic risk.
\end{enumerate}

For Step 1, one may use the model selection criterion such as AIC, BIC or best subset selection. We, however, use BIC. In Table \ref{Tab:fittings}, we show the full and candidate sub-model.

\begin{table}[!htbp]
\centering
\begin{tabular}{lll}
\toprule
\textbf{Models} & \textbf{Formulas} \\
\midrule
Full Model &  rmort = $\beta_0+\beta_1$tempr $+\beta_2$rh$+\beta_3$co$+\beta_4$so2$+\beta_5$no2$+\beta_6$hycarb$+\beta_7$o3 +$\beta_8$part \\
Sub-Model &  rmort = $\beta_0+\beta_1$tempr $+\beta_2$co\\
\bottomrule
\end{tabular}
\caption{The full and candidate sub-model
\label{Tab:fittings}}
\end{table}


\begin{figure}[!htbp]
\centering
   \includegraphics[height=8cm,width=12cm]{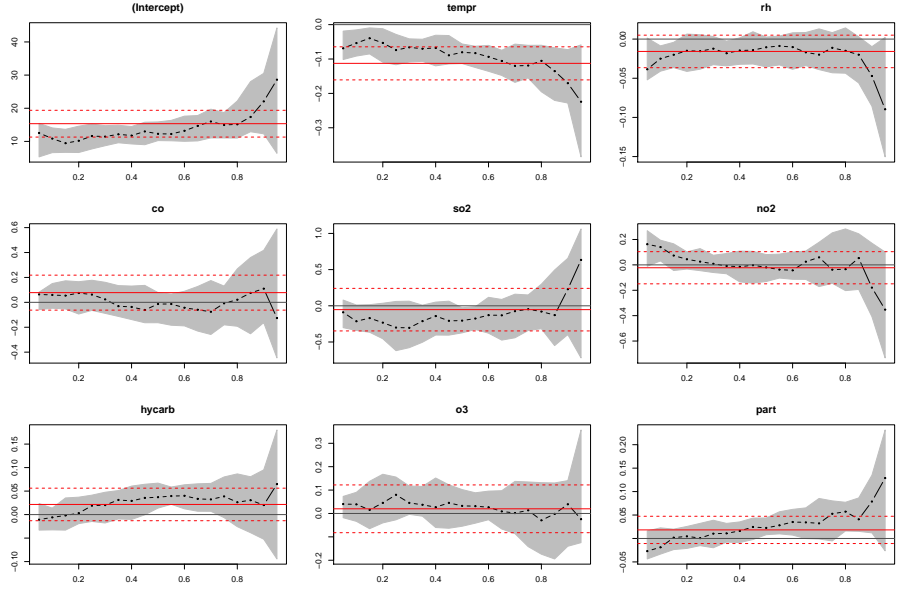}
   \caption{OLS and the full model Quantile Regression Estimates for
    LA Pollution-Mortality data set}
   \label{fig:coef} 
\end{figure}

Figure \ref{fig:coef} presents a summary of the OLS and the FM quantile regression results.
Here, we have 8 covariates, plus an intercept. For each of the 9
coefficients, we plot the 19 distinct quantile regression estimates for $\tau$ ranging from
$0.05$ to $0.95$ as the solid curve with filled dots. For each covariate, these point
estimates may be interpreted as the impact of a one-unit change of the covariate on the response variable
respiratory mortality other covariates fixed. Thus, each of the plots has a horizontal
quantile, or $\tau$, scale, and the vertical axes indicates the covariate effect.
The solid line in each figure shows the OLS estimate of the
conditional mean effect. The two dotted lines represent conventional 90 percent
confidence intervals for the OLS estimate. The shaded gray area depicts a
90 percent point-wise confidence band for the quantile regression estimates.

We will confine our discussion as follows: The intercept estimates seem more dependent on the particular quantile. For example, up to the third quantile, quantile estimates are lower than the OLS while it is larger than the OLS for the upper quantile. With the exception of the coefficients co, hycarb and o3, the quantile regression estimates lie at some point outside the confidence intervals for the OLS regression, suggesting that the effects of these covariates may change across the conditional distribution of the independent variable.

In order to analyze this example, we bootstrap the data using 1000 resamplings. After that, we split the data into two which are training and test data sets. Furthermore, we center the co-variates of training and test data set based on the training data set independently. Finally, we computed the predictive mean absolute deviation (PMAD) criterion which is defined by
\begin{equation*}
\rm PMAD(\hbbeta_{\tau}^{\ast}) = \frac{1}{n_{test}}\sum_{i=1}^{n_{test}}\left | \by_{test}-\bX_{test}\hbbeta_{\tau}^{\ast} \right |.
\end{equation*}

We evaluate the performance of the estimators by averaged cross validation (CV) error using a 5-fold CV. In Table \ref{Tab:PMAD}, we report the performance of the estimator in the sense of PMAD for the real data application. As expected, the SM estimator has the lowest PMAD value for all $\tau$ values. The PS performs better than the Lasso, Elastic-net, FM and OLS, especially in the first and second quantile (median), while the Ridge outperforms all others since the data has highly the problem of multicollinearity. Also, the performance of PT is also well in median.

\begin{table}[!htbp]
\centering
\begin{tabular}{rrrrrrrr}
  \toprule
$\tau$ & FM & SM & PT & PS & Ridge & Lasso & ENET \\ 
  \midrule
0.25 & 2.612 & 2.209 & 2.612 & 2.515 & 2.337 & 2.531 & 2.535 \\ 
  0.5 & 2.391 & 1.881 & 2.359 & 2.249 & 2.220 & 2.372 & 2.372 \\ 
  0.75 & 3.082 & 2.041 & 3.063 & 2.802 & 2.275 & 2.557 & 2.469 \\ 
   \midrule
  & OLS \\
Mean& 2.803\\  
\bottomrule
\end{tabular}
\caption{The PMAD values of the listed estimations}
\label{Tab:PMAD}
\end{table}

\section{Simulation}
\label{sim}
We conduct Monte-Carlo simulation experiments to study the performances of the proposed estimators under various practical settings. In order to generate the response variable, we use
\begin{equation*}
y_i=\bx_i'\bbeta+\varepsilon_i,\ i=1,\dots,n,
\end{equation*}%
where $\bx_i$'s are standard normal. The correlation between the $j$th and $k$th components of $\bX$ equals to $0.5^{|j-k|}$ and also $\varepsilon_i$'s are dependent. 

We consider $\bbeta' = (3, 1.5, 0, 0, 2, 0, 0, 0)$. Also, we simulate data which contains a training dataset, validation set and an independent test set. Note that the co-variates are scaled to have mean zero and unit variance. We fitted the models only using the training data and the tuning parameters were selected using the validation data.  We also use the notation $\cdot/\cdot/\cdot$ to describe the number of observations in the training, validation and test sets respectively. Hence, we consider that the each data set consists of $50/50/200$ observations and $\bX \sim N(0,\boldsymbol{\Sigma})$, where $\Sigma_{ij}=0.5^{|i-j|}$. Furthermore, the errors follow AR(1) process, that is, 
\begin{equation*}
\varepsilon_i = \rho\varepsilon_{i-1}+\omega_t
\end{equation*}
where $|\rho|<1$ is called the ``autocorrelation parameter" and the $\omega_t$ term is a new error term that follows the usual regression assumptions: $\omega_t \sim_{iid} \mathcal{N}(0,1)$.

\begin{table}[!htbp]
\centering
\begin{tabular}{rrcccc}
  \toprule
$\tau$ & & $\rho=-0.2$ & $\rho=0.2$ & $\rho=-0.5$ & $\rho=0.5$ \\ 
  \midrule
0.25& FM & 0.190(0.004) & 0.188(0.004) & 0.206(0.005) & 0.220(0.004) \\ 
  &SM & 0.060(0.002) & 0.057(0.002) & 0.069(0.002) & 0.063(0.002) \\ 
  &PT & 0.078(0.005) & 0.076(0.006) & 0.090(0.007) & 0.083(0.006) \\ 
  &PS & 0.120(0.004) & 0.134(0.005) & 0.143(0.005) & 0.149(0.005) \\ 
  \cmidrule(lr){3-6}
  &Ridge & 0.135(0.003) & 0.139(0.004) & 0.154(0.003) & 0.158(0.003) \\ 
  &Lasso & 0.081(0.002) & 0.076(0.003) & 0.087(0.003) & 0.089(0.002) \\ 
  &ENET & 0.078(0.002) & 0.074(0.003) & 0.085(0.003) & 0.087(0.002) \\ 
     \midrule
0.5&FM & 0.173(0.004) & 0.165(0.003) & 0.199(0.005) & 0.191(0.004) \\ 
  &SM & 0.055(0.002) & 0.053(0.002) & 0.058(0.002) & 0.057(0.002) \\ 
  &PT & 0.059(0.005) & 0.061(0.004) & 0.072(0.006) & 0.069(0.006) \\ 
  &PS & 0.103(0.004) & 0.098(0.004) & 0.126(0.005) & 0.122(0.005) \\ 
  \cmidrule(lr){3-6}
  &Ridge & 0.133(0.003) & 0.124(0.003) & 0.150(0.003) & 0.149(0.003) \\ 
  &Lasso & 0.073(0.002) & 0.073(0.002) & 0.078(0.003) & 0.077(0.003) \\ 
  &ENET & 0.072(0.002) & 0.072(0.002) & 0.078(0.003) & 0.075(0.003) \\ 
    \midrule
0.75&FM & 0.183(0.004) & 0.184(0.004) & 0.217(0.005) & 0.210(0.005) \\ 
  &SM & 0.060(0.002) & 0.059(0.002) & 0.062(0.002) & 0.066(0.002) \\ 
  &PT & 0.082(0.005) & 0.072(0.005) & 0.080(0.006) & 0.090(0.006) \\ 
  &PS & 0.121(0.004) & 0.115(0.004) & 0.146(0.005) & 0.144(0.005) \\ 
  \cmidrule(lr){3-6}
  &Ridge & 0.140(0.003) & 0.139(0.003) & 0.161(0.004) & 0.159(0.004) \\ 
  &Lasso & 0.080(0.002) & 0.074(0.002) & 0.087(0.003) & 0.089(0.003) \\ 
  &ENET & 0.078(0.002) & 0.073(0.002) & 0.085(0.003) & 0.083(0.003) \\ 
   \midrule
Mean&OLS & 0.137(0.009) & 0.136(0.009) & 0.154(0.010) & 0.154(0.010) \\
   \bottomrule
\end{tabular}
\caption{Simulated PMAD values of estimators, and the values in parenthesis present the standard errors of each estimation}
\label{TAb:sim}
\end{table}

Table \ref{TAb:sim} presents an outline summary for the different illustrative models used in the case of autoregressive errors  where  $\rho=\pm5$ characterized by heavier tails while $\rho=\pm2$ corresponds to the median. First, we note that the OLS fails against to quantile-type estimations. As expected, the SM has the lowest the PMAD value since the data is generated from an empirical distribution where the candidate sub-model is nearly true. Furthermore, the pretest and positive shrinkage estimators are superior to the FM estimator. On the other hand, the results indicate that the PT mostly performs better than penalty estimators while positive shrinkage does not have a good performance due to the small value of $p_1$.

\section{Theoretical Results}
\label{sec:TR}
In this section, we demonstrate the asymptotic risk properties of suggested estimators. So, we consider the following theorem. 

\begin{theorem}\label{teo_dist_FM_AR(1)} The distribution of quantile regression model with AR(1) process is given by
\begin{equation}
\sqrt{n}(\hbbeta_\tau-\bbeta)\rightarrow^{\hspace{-0.3cm}D}\mathcal{N}(0,w^2\boldsymbol{\Gamma})
\label{teo_dist_FM}
\end{equation}
\end{theorem}
\begin{proof}
The proof can be obtained from (\citet{Davino2014})
\end{proof}
Let a sequence of local alternatives $\left\{K_n\right\}$ given by
\begin{equation*}
\ K_n:\boldsymbol{\beta}_{2,\tau}=\frac{\boldsymbol{\gamma}}{\sqrt{n}}
\end{equation*}
where $\boldsymbol{\gamma}=\left(
\gamma_{1},\gamma_{2},\dots,\gamma_{p_{2}}\right)'\in \Re^{p_{2}}$ is a fixed vector. If $\boldsymbol{\gamma}=\bold{0}_{p_2}$, then the null hypothesis is true. Furthermore, we consider the following proposition to establish the
asymptotic properties of the estimators.

\begin{proposition}\label{prop_vector_dist} 
Let $\vartheta _{1} = \sqrt{n}\left(\hbbeta_{1, \tau}-%
\bbeta_{1,\tau}\right)$, $\vartheta _{2} =\sqrt{n}\left( \tbbeta_{1,\tau}-%
\bbeta_{1,\tau}\right)$ and $\vartheta _{3} =\sqrt{n}\left( \hbbeta_{1,\tau}-%
\tbbeta_{1,\tau}\right)$. Under the regularity assumptions (i)--(iii), Theorem~\ref{teo_dist_FM_AR(1)} and the local alternatives $\left\{ K_{n}\right\}$, as $n\rightarrow \infty$ we have the joint distributions are given as follows:

$$\left(
\begin{array}{c}
\vartheta _{1} \\
\vartheta _{3}%
\end{array}%
\right) \sim\mathcal{N}\left[ \left(
\begin{array}{c}
\boldsymbol{0 }_{p_1} \\
-\boldsymbol{\delta }%
\end{array}%
\right) ,\left(
\begin{array}{cc}
\omega^2\bGamma_{11.2}^{-1} & \bSigma_{12} \\
\bSigma_{21} & \boldsymbol{\Phi}%
\end{array}%
\right) \right],$$


$$\left(
\begin{array}{c}
\vartheta _{3} \\
\vartheta _{2}%
\end{array}%
\right) \sim\mathcal{N}\left[ \left(
\begin{array}{c}
-\boldsymbol{\delta } \\
\boldsymbol{\delta }%
\end{array}%
\right) ,\left(
\begin{array}{cc}
\boldsymbol{\Phi} & \bSigma^* \\
\bSigma^* & \omega^2\bGamma_{11}^{-1}%
\end{array}%
\right) \right],$$ \\
where $\bdelta= \bGamma_{11}^{-1}\bGamma_{12}\boldsymbol{\gamma}$,  $\boldsymbol{\Phi}=\bGamma_{11}^{-1}\bGamma_{12}\bGamma_{22.1}^{-1}
\bGamma_{21}\bGamma_{11}^{-1}$, $\bSigma_{12}=-\bGamma_{12}\bGamma_{21}\bGamma_{11}^{-1}$ and $\bSigma^*=\bSigma_{21}+\omega^2\bGamma_{11.2}^{-1}$.
\end{proposition}

Now, we are ready to obtain the asymptotic distributional risks of estimators which are given the following section.

\subsection{The performance of Risk}
The asymptotic distributional risk of an estimator $\hbbeta_{1, \tau}^{\ast }$ is defined as
\begin{equation*}
\mathcal{R}\left( \hbbeta_{1,\tau}^{\ast} \right) = {\rm tr}\left( \boldsymbol{W\Gamma} \right) 
\end{equation*}
where $\boldsymbol{W}$ is a positive definite matrix of weights with dimensions of $p\times p$, and $\boldsymbol{\Gamma}$ is the asymptotic covariance matrix of an estimator $\hbbeta_{1, \tau}^{\ast }$ is defined as
\begin{equation*}
\boldsymbol{\Gamma} \left( \hbbeta_{1,\tau}^{\ast} \right) = \mathbb{E}\left\{\underset{n\rightarrow \infty }{\lim }{n}\left( \hbbeta_{1,\tau}^{\ast} -\bbeta_{1,\tau}\right)\left( \hbbeta_{1,\tau}^{\ast}-\bbeta_{1,\tau}\right) ^{'}\right\}.
\end{equation*}

\begin{theorem}\label{risks}
Under the assumed regularity conditions in (i) and (ii), the Proposition \ref{prop_vector_dist},  the Theorem~\ref{teo_dist_FM}  and $\left\{K_n\right\}$, the expressions for asymptotic risks for listed estimators are:
\begin{eqnarray*}
\mathcal{R}\left( \hbbeta_{1,\tau}\right)
&=&\omega^2{\rm tr}\left( \bW \bGamma_{11.2}^{-1}\right)\\
\mathcal{R}\left( \tbbeta_{1,\tau}\right)
&=&\omega^2 {\rm tr}\left( \bW \bGamma_{11}^{-1}\right)
+{\rm tr}\left(\bW\bM\right), where \left(\bM=\bGamma_{11}^{-1}\bGamma_{12}\boldsymbol{\gamma}\boldsymbol{\gamma}^{'}\bGamma_{21}\bGamma_{11}^{-1}=\bdelta\bdelta^{'}\right)\\
\mathcal{R}\left( \hbbeta_{1,\tau}^{\rm PT}\right) 
&=&\mathcal{R}\left(
\hbbeta_{1,\tau}\right)+\omega^2{\rm tr}\left(\bW\bGamma_{11}^{-1}\bGamma_{12}\bGamma_{22.1}^{-1}\bGamma_{21}\bGamma_{11}^{-1}\right)+{\rm tr}\left(\bdelta\bW\bdelta^{'}\right)\mathbb{H}_{d+4}\left( \chi_{d+4}^{2};\Delta \right) \\
&&+\bW\boldsymbol{\Phi}\mathbb{H}_{d+4}\left(\chi _{d+2,\alpha }^{2}\left(\Delta \right)\right)+{\rm tr}\left(\bdelta\bW\bdelta^{'}\right)\mathbb{H}_{d+6}\left( \chi_{d+2,\alpha }^{2}\left(\Delta\right) \right)  \cr     
\mathcal{R}\left( \hbbeta_{1,\tau}^{\rm S}\right) &=&\mathcal{R}\left(
\hbbeta_{1,\tau}\right) -2d\mathbb{E}\left\{ \chi _{d+4
}^{-2}\left( \Delta \right) \right\}{\rm tr}\left(\bW\bSigma_{21}\right)  \cr
&&-2d \mathbb{E}\left\{\chi _{d+6
}^{-2}\left( \Delta\right)\right\}{\rm tr}\left(\bW\bdelta\bdelta^{'}\bSigma^{*-1}\bSigma_{21}\right) \cr
&&+2d \mathbb{E}\left\{\chi _{d+4
}^{-2}\left( \Delta\right)\right\}{\rm tr}\left(\bW\bdelta\bdelta^{'}\bSigma^{*-1}\bSigma_{21}\right) \cr
&&+d^2\mathbb{E}\left\{\chi _{d+4
}^{-4}\left( \Delta\right)\right\}{\rm tr}\left(\bW\bSigma^{*}\right) \cr
&&+d^2\mathbb{E}\left\{\chi _{d+6
}^{-2}\left( \Delta\right)\right\}{\rm tr}\left(\bW\bdelta\bdelta^{'}\right)\cr
\mathcal{R}\left( \hbbeta_{\tau}^{\rm PS}\right)
&=&\mathcal{R}\left(
\hbbeta_{\tau}^{\rm S}\right)-2\mathbb{E}\left(1-d\chi _{d+4 }^{-2}\left( \Delta \right)\right)\textrm{I}\left( \chi _{d+4 }^{2}\left( \Delta \right)<d\right) {\rm tr}\left(\bW\bSigma_{21}\right) \cr
&&-2\mathbb{E}\left(1-d\chi _{d+6 }^{-2}\left( \Delta \right)\right)\textrm{I}\left(\chi_{d+6}^{2}\left( \Delta \right)<d\right){\rm tr}\left(\bW\bdelta^{'}\bdelta\bSigma^{*-1}\bSigma_{21}\bSigma^{*-1}\bSigma_{21}\right) \cr
&&+2\mathbb{E}\left(1-d\chi _{d+4 }^{-2}\left( \Delta \right)\right)\textrm{I}\left(\chi_{d+6}^{2}\left( \Delta \right)<d\right){\rm tr}\left(\bW\bdelta\bdelta^{'}\right) \cr
&&+\mathbb{E}\left(1-d\chi _{d+4 }^{-2}\left( \Delta \right)\right)^2\textrm{I}\left(\chi_{d+4}^{2}\left( \Delta \right)<d\right){\rm tr}\left(\bW\bSigma^{*}\right) \cr
&&+\mathbb{E}\left(1-d\chi _{d+6 }^{-2}\left( \Delta \right)\right)^2\textrm{I}\left(\chi_{d+6}^{2}\left( \Delta \right)<d\right){\rm tr}\left(\bW\bdelta\bdelta^{'}\right) 
\end{eqnarray*}%
\end{theorem}
Noting that if $\bGamma_{12}= \boldsymbol{0}$, then all the risks reduce to common value  $\omega^2{\rm tr}\left( \bW \bGamma_{11}^{-1}\right)$ for all $\bW$. For  $\bGamma_{12}\neq 0$, the risk of  $\hbbeta_{1,\tau}$ remains constant while the risk of $\tbbeta_{1,\tau}$ is an bounded function of $\Delta$ since $\Delta \in [0,\infty]$. The risk of $\hbbeta_{1,\tau}^{\rm PT}$ increases as $\Delta$ moves away from zero, achieves it maximum and then decreases towards the risk of the full model estimator. Thus, it is a bounded function of $\Delta$. The risk of $\hbbeta_{1,\tau}$ is smaller than the risk of $\hbbeta_{1,\tau}^{\rm PT}$ for some small values of $\Delta$ and opposite conclusions holds for rest of the parameter space. It can be seen that 
\begin{equation*}
\mathcal{R}\left(\hbbeta_{1,\tau}^{\rm PS}\right)\leq{\mathcal{R}\left(
\hbbeta_{1,\tau}^{\rm S}\right)}\leq{\mathcal{R}\left(\hbbeta_{1,\tau}\right)},
\end{equation*}
strictly inequality holds for small values of $\Delta$. Thus, positive shrinkage is superior to the shrinkage estimator. However, both shrinkage estimators outperform the full model estimator in the entire parameter space induced by $\Delta$. On the other hand, the pretest estimator performs better than the shrinkage estimators when $\Delta$ takes small values and outside this interval the opposite conclusion holds.

\section{Conclusions}
\label{conc}

In this paper, we obtained pretest and stein-type shrinkage estimations based on quantile regression when the distribution of errors have the problem of autocorrelation. Also, we investigated the performance of the listed estimators in a real world example using the data analyzed by \cite{Shumway1998} such that the effects of air pollution and temperature on weekly mortality in LA are considered. The results showed that the quantile type estimators outperform the OLS. Not surprisingly, the SM estimator has the lowest PMAD since the candidate sub-model is assumed as true. Furthermore, the PT and PS perform better than the FM. Also, the performance of the proposed estimators are mostly superior to penalty estimators, moreover Ridge has a better performance since the data has the multicollinearity problem. On the other hand, we conducted a Monte Carlo simulation study in order to investigate the performance of the suggested estimators. The results of simulation study coincide with the results of real data example. Finally, we demonstrated the asymptotic distributional risk performance of the listed estimators. Our asymptotic theory is well supported by numerical analysis.


\end{document}